\documentclass[copyright,creativecommons,noderivs,noncommercial]{eptcs}

\usepackage{svn-multi}
\usepackage{amsmath}
\usepackage{amssymb}
\usepackage{amsthm}
\usepackage{graphicx}
\usepackage{xspace,stmaryrd}
\usepackage{datetime}
\usepackage{etoolbox}
\usepackage{longtable}
\usepackage{ulem}
\usepackage{color}

\DeclareSymbolFont{letters}{OML}{txmi}{m}{it}

\newcommand{\Vp}{V_{{P}}}
\newcommand{\Vd}{V_{{D}}}
\newcommand{\Sigmap}{\Sigma_{{P}}}
\newcommand{\Sigmad}{\Sigma_{{D}}}
\newcommand{\State}[2]{( #1, #2 )}
\newcommand{\sdspace}{\textnormal{$x_{D}$}}
\newcommand{\ssdspace}{\textnormal{$y_{D}$}}

\newcommand{\nr}[1]{(#1) \hspace*{0.1cm}}
\newcommand{\Seq}{\mathrel{;}}
\newcommand{\cl}{\textit{cl}}

\newcommand{\novel}[1]{{\textcolor{blue}{#1}}}

\newenvironment{novelty}{\color{blue}}{}

\newcommand{\chk}{\textit{check}}
\newcommand{\upd}{\textit{update}}

\newcommand{\ask}{\textit{ask}}
\newcommand{\nask}{\textit{nask}}
\newcommand{\tell}{\textit{tell}}
\newcommand{\get}{\textit{get}}

\newcommand{\BCCSP}{\textnormal{\textnormal{BCCSP}}}
\newcommand{\BCCSPD}{\textnormal{\textnormal{BCCSP}$_{D}$}}
\newcommand{\BCCSPDc}{\textnormal{\textnormal{BCCSP}$_{D}^{c}$}}

\DeclareSymbolFont{operators}{OT1}{cmr}{m}{n}
\DeclareSymbolFont{letters}{OML}{cmm}{m}{it}
\DeclareSymbolFont{symbols}{OMS}{cmsy}{m}{n}
\DeclareSymbolFont{largesymbols}{OMX}{cmex}{m}{n}

\svnid{$Id: paper.tex 1625 2013-07-20 15:06:46Z eugen $}

\settimeformat{hhmmsstime}
\newcommand{\builddate}{\ifnumcomp{\year}{<}{10}{0}{}\the\year-\ifnumcomp{\month}{<}{10}{0}{}\the\month-\ifnumcomp{\day}{<}{10}{0}{}\the\day\ \currenttime}
\ifnumcomp{\svnhour}{=}{23}{\def\svnhour{01}}{}
\ifnumcomp{\svnhour}{=}{22}{\def\svnhour{00}}{}
\ifnumcomp{\svnhour}{=}{21}{\def\svnhour{23}}{}
\ifnumcomp{\svnhour}{=}{20}{\def\svnhour{22}}{}
\ifnumcomp{\svnhour}{=}{19}{\def\svnhour{21}}{}
\ifnumcomp{\svnhour}{=}{18}{\def\svnhour{20}}{}
\ifnumcomp{\svnhour}{=}{17}{\def\svnhour{19}}{}
\ifnumcomp{\svnhour}{=}{16}{\def\svnhour{18}}{}
\ifnumcomp{\svnhour}{=}{15}{\def\svnhour{17}}{}
\ifnumcomp{\svnhour}{=}{14}{\def\svnhour{16}}{}
\ifnumcomp{\svnhour}{=}{13}{\def\svnhour{15}}{}
\ifnumcomp{\svnhour}{=}{12}{\def\svnhour{14}}{}
\ifnumcomp{\svnhour}{=}{11}{\def\svnhour{13}}{}
\ifnumcomp{\svnhour}{=}{10}{\def\svnhour{12}}{}
\ifnumcomp{\svnhour}{=}{09}{\def\svnhour{11}}{}
\ifnumcomp{\svnhour}{=}{08}{\def\svnhour{10}}{}
\ifnumcomp{\svnhour}{=}{07}{\def\svnhour{09}}{}
\ifnumcomp{\svnhour}{=}{06}{\def\svnhour{08}}{}
\ifnumcomp{\svnhour}{=}{05}{\def\svnhour{07}}{}
\ifnumcomp{\svnhour}{=}{04}{\def\svnhour{06}}{}
\ifnumcomp{\svnhour}{=}{03}{\def\svnhour{05}}{}
\ifnumcomp{\svnhour}{=}{02}{\def\svnhour{04}}{}
\ifnumcomp{\svnhour}{=}{01}{\def\svnhour{03}}{}
\ifnumcomp{\svnhour}{=}{00}{\def\svnhour{02}}{}
\newcommand{\wipinfo}{
  \ \ Last change: \svnyear-\svnmonth-\svnday\ \svnhour:\svnminute:\svnsecond \\
  \qquad\quad Build: \builddate \\
  \quad\qquad Revision: \svnfilerev \qquad\qquad\qquad\qquad\ \
}

\newtheorem{defn}{Definition}
\newtheorem{lem}[defn]{Lemma}
\newtheorem{thm}[defn]{Theorem}

\newcommand{\commform}{\textsf{comm-form}}

\newcommand{\COMM}{\id{COMM}}
\newcommand{\simcom}{\mathrel{\sim_{cc}}}

\newcommand{\sosrule}[2]{\frac{\raisebox{.7ex}{\normalsize{$#1$}}}
                        {\raisebox{-1.0ex}{\normalsize{$#2$}}}}
\newcommand{\trans}[1]{\,{\xrightarrow{{#1}}}\,}

\newcommand{\gtrans}[2]{\,{\stackrel{{#2}}{\rightarrow_{#1}}}\,}

\newcommand{\bisimpar}[1]{\;\underline{\hspace*{-0.15ex}
                        \leftrightarrow\hspace*{-0.15ex}}^{{#1}}}
\newcommand{\bisimsbpar}[1]{\;{\hspace*{-0.15ex}
                        \underline{\leftrightarrow}_{\textnormal{sl}}^{{#1}}\hspace*{-0.15ex}}\;}

\newcommand{\Rsb}{R_{\textnormal{sl}}}

\newcommand{\DedRule}[1]{\mbox{{\bf\small (#1)}}}

\newcommand{\id}[1]{\mathit{#1}}

\def\lparal{\mathbin{\setbox0=\hbox{$\|$}%
        \dimen0=\dp0 \advance\dimen0 -1.5pt \dp0=\dimen0%
        \underline{\kern-1.5pt\box0\kern1.5pt}}}

\newcommand{\Par}{\mathrel{||}}

\newcommand{\nul}{{\ensuremath{\mathbf{0}}}}

\newcommand{\communication}[2]{\ensuremath{#1 \, \| \, #2}}
\newcommand{\tuple}[1]{\ensuremath{\left( \, #1 \, \right)}}

\newcommand{\terminate}[1]{\ensuremath{\tuple{#1} \, \checkmark}}
\newcommand{\terminatee}[2]{\ensuremath{{#1} \,\,\, \checkmark_{#2}}}
\newcommand{\sequential}[2]{\ensuremath{#1 \, \odot \, #2}}
\newcommand{\choice}[2]{\ensuremath{#1 \, \oplus \, #2}}
\newcommand{\disrupt}[2]{\ensuremath{ #1 \, \blacktriangleright \, #2 }}
\newcommand{\encapsulate}[2]{\ensuremath{\partial_{#1} \left( #2 \right) }}
\newcommand{\init}[2]{\ensuremath{#1 \gg #2}}
\def\lparal{\mathbin{\setbox0=\hbox{$\|$}%
        \dimen0=\dp0 \advance\dimen0 -1.5pt \dp0=\dimen0%
                \underline{\kern-1.5pt\box0\kern1.5pt}}}
\newcommand{\lcommunication}[2]{#1 \, \mathbin{\lparal} \, #2}
\newcommand{\fcommunication}[2]{\ensuremath{#1 \, | \, #2}}
\newcommand{\ldisrupt}[2]{\ensuremath{ #1 \, \vartriangleright \, #2 }}
\newcommand{\modelsf}{\ensuremath{\models_{\mathrm{f}}}}
\newcommand{\modelsr}{\ensuremath{\models_{\mathrm{r}}}}
\newcommand{\Dom}[1]{\ensuremath{\mathit{dom}(#1)}}
\newcommand{\clinterval}[2]{[ #1 , #2]}
\newcommand{\ctrans}[3]{\ensuremath{\tuple{#1} \, \stackrel{#2}{\leadsto} \, \tuple{#3}}}
\newcommand{\ctranss}[3]{\ensuremath{{#1} \, \stackrel{#2}{\leadsto} \, {#3}}}

\newcommand{\bbn}{|\mkern-2mu[}
\newlength{\chisymbolwidth}
\definecolor{lightblue}{RGB}{224,224,255}
\definecolor{lightred}{RGB}{255,224,224}
\definecolor{lightgreen}{RGB}{224,255,224}
\definecolor{lightyellow}{RGB}{255,255,224}
\definecolor{lightpurple}{RGB}{255,224,255}
\definecolor{darkerred}{RGB}{64,0,0}
\definecolor{darkred}{RGB}{128,0,0}
\definecolor{darkblue}{RGB}{0,0,128}
\definecolor{darkgreen}{RGB}{0,128,0}
\definecolor{darkpurple}{RGB}{128,0,128}

\newcommand{\colorpar}[3]{\colorbox{#1}{\parbox{#2}{#3}}}
\newcommand{\marginremark}[3]{\marginpar{\colorpar{#2}{\linewidth}{\color{#1}#3}}}

\makeatletter
\def\THICKhrulefill{\leavevmode \leaders \hrule height 5pt\hfill \kern \z@}
\makeatother

\newcommand{\remarkDG}[1]{\marginremark{darkred}{lightred}{\tiny{[DG]~ #1}}}
\newcommand{\remarkWF}[1]{\marginremark{darkgreen}{lightgreen}{\tiny{[WF]~ #1}}}
\newcommand{\remarkMM}[1]{\marginremark{darkpurple}{lightpurple}{\tiny{[MM]~ #1}}}
\newcommand{\remarkEG}[1]{\marginremark{darkblue}{lightblue}{\tiny{[EG]~ #1}}}

\newtheorem{remark}[defn]{Remark}

\newif\ifdraft
\draftfalse

\ifdraft
\else
\renewcommand{\remarkDG}[1]{}
\renewcommand{\remarkWF}[1]{}
\renewcommand{\remarkMM}[1]{}
\renewcommand{\remarkEG}[1]{}
\fi

\begin{document}

\title{Algebraic Meta-Theory of Processes with Data}

\author{
Daniel Gebler
\institute{Department of Computer Science, VU University Amsterdam (VU),
\\ De Boelelaan 1081a, NL-1081~HV~Amsterdam, The Netherlands
}
\and
Eugen-Ioan Goriac
\institute{ICE-TCS, School of Computer Science, Reykjavik University,
\\ Menntavegur 1, IS-101, Reykjavik, Iceland
}
\and
Mohammad Reza Mousavi
\institute{Center for Research on Embedded Systems (CERES),
Halmstad University\\ Kristian IV:s v\"ag 3, SE-302 50, Halmstad, Sweden
}
\ifdraft
\and\institute{\wipinfo}
\fi
}

\def\authorrunning{Gebler, Goriac \& Mousavi}
\def\titlerunning{Algebraic Meta-Theory of Processes with Data}

\maketitle

\begin{abstract}
There exists a rich literature of rule formats guaranteeing different algebraic properties for formalisms with a Structural Operational Semantics. Moreover, there exist a few approaches for automatically deriving axiomatizations characterizing strong bisimilarity of processes.
To our knowledge, this literature has never been extended to the setting with data (e.g. to model storage and memory).
We show how the rule formats for algebraic properties can be exploited in a generic manner in the setting with data.
Moreover, we introduce a new approach for deriving sound and ground-complete axiom schemata for a notion of bisimilarity with data, called stateless bisimilarity, based on intuitive auxiliary function symbols for handling the store component. We do restrict, however, the axiomatization to the setting where the store component is only given in terms of constants.
\end{abstract}

\section{Introduction}\label{Sect:intro}

Algebraic properties capture some key features of programming and specification constructs
and can be used both as design principles (for the semantics of such constructs) as well as for verification of
programs and specifications built using them.
{When given the semantics of a language, inferring} properties such as commutativity, associativity and unit element, {as well deriving sets of axioms for reasoning on the behavioural equivalence of two processes} constitute one of the cornerstones of process algebras \cite{Baeten05,Roscoe10} and play essential roles in several disciplines for behavioural modeling and analysis such as term rewriting \cite{Baader99} and model checking \cite{Baier08}.

For formalisms with a Structural Operational Semantics (SOS), there exists a rich literature on
meta-theorems guaranteeing key algebraic properties (commutativity \cite{Mousavi05-IPL}, associativity \cite{Mousavi08-CONCUR}, zero  and unit elements \cite{Mousavi11-TCS}, idempotence \cite{Mousavi12-SCICO}, and distributivity \cite{Mousavi11-LATA}) by means of restrictions on the syntactic shape of the transition rules. {At the same time, for GSOS \cite{Bloom95}, a restricted yet expressive form of SOS specifications, one can obtain a sound and ground-complete axiomatization modulo strong bisimilarity \cite{Aceto94}.}
Supporting some form of data (memory or store) is a missing aspect of these existing meta-theorems,
which bars applicability to the semantics of numerous programming languages and formalisms that do feature these aspects in different forms.

In this paper we provide a natural and generic link between the meta-theory of algebraic properties and axiomatizations, and SOS with data for which we consider that one data state models the whole memory.
Namely, we move the data terms in SOS with data to the labels and instantiate them to closed terms; we call this process \emph{currying}.
Currying allows us to apply directly the existing rule formats for algebraic properties on the curried SOS specifications (which have process terms as states and triples of the form (datum, label, datum) as labels). {We also present a new way of automatically deriving sound and ground-complete axiomatization schemas modulo strong bisimilarity for the curried systems for the setting in which the data component is characterized by constants.}
It turns out that strong bisimilarity for the curried SOS specification coincides with the notion of stateless bisimilarity in the original SOS specifications with data. The latter notion is extensively studied in \cite{Mousavi05-IC} and used, among others, in \cite{Bergstra07,Galpin12,Beek06,Beek07}.
(This notion, in fact, coincides with the notion of strong bisimilarity proposed for Modular SOS in \cite[Section 4.1]{Mosses04a}.)
Hence, using the existing rule formats, we can obtain algebraic laws for SOS specification with data that are sound with respect to stateless bisimilarity,
as well as the weaker notions of initially stateless bisimilarity and statebased bisimilarity, studied in \cite{Mousavi05-IC}.

\paragraph{Related work.} SOS with data and store has been extensively used in specifying semantics
of programming and specification languages, dating back to the original work of Plotkin \cite{Plotkin04b,Plotkin04a}.
Since then, several pieces of work have been dedicated to providing a formalization for SOS specification frameworks allowing one to include data and store and reason over it.
The current paper builds upon the approach proposed in \cite{Mousavi05-IC} (originally published as \cite{Mousavi04-LICS}).

The idea of moving data from the configurations (states) of operational semantics to labels is reminiscent of  Modular SOS \cite{Mosses04b,Mosses04a}, Enhanced SOS \cite{Degano01}, the Tile Model \cite{Gadducci00}, and context-dependent-behaviour framework of \cite{Colvin11}. The idea has also been applied in instances of SOS specification, such as those reported in \cite{Baeten97,Beek06,Owens08}. The present paper contributes to this body of knowledge by presenting a generic transformation from SOS specifications with data and store (as part of the configuration) to Transition System Specifications \cite{Bloom95,Groote92}. The main purpose of this generic transformation is to enable exploiting the several existing rule formats defined on transition system specifications on the results of the transformation and then, transform the results back to the original SOS specifications (with data and store in the configuration) using a meaningful and well-studied notion of bisimilarity with data.
Our transformation is also inspired by the translation of SOS specifications of programming languages into rewriting logic, see e.g., \cite{Meseguer02,Meseguer04a}.

\paragraph{Structure of the paper.}
The rest of this paper is organized as follows.
In Section \ref{sec:pre}, we recall some basic definitions regarding SOS specifications and behavioural equivalences.
In Section \ref{sec:currying}, we present the currying technique and formulate the theorem regarding the correspondence between strong and stateless bisimilarity. {In Section~\ref{sec:axiomatizing} we show how to obtain sound and ground-complete axiomatizations modulo strong bisimilarity for those curried systems for which the domain of the data component is a finite set of constants.}
We apply the currying technique to Linda \cite{Carriero89}, a coordination language from the literature chosen as case study in Section \ref{sec:cases}, and show how key algebraic properties of the operators defined in the language semantics are derived.
We conclude the paper in Section \ref{sec:conc}, by summarizing the results and presenting some  directions for future work.

\section{Preliminaries}\label{sec:pre}

\subsection{Transition Systems Specifications}

We assume a multisorted signature $\Sigma$ with designated and distinct sorts $P$ and $D$ for processes and data, respectively.
Moreover, we assume infinite and disjoint sets of process variables $\Vp$ (typical members: $x_{P}, y_{P}, x_{P_{i}}, y_{P_{i}} \ldots$) and data variables $\Vd$ (typical members: $x_{D}, y_{D}, x_{D_{i}}, y_{D_{i}} \ldots$), ranging over their respective sorts $P$ and $D$.

Process and data signatures, denoted respectively by $\Sigmap \subseteq \Sigma$ and $\Sigmad \subseteq \Sigma$, are sets of function symbols with fixed arities.
We assume in the remainder that the function symbols in $\Sigmad$ take only parameters of the sort $\Sigmad$, while those in $\Sigmap$ can take parameters both from $\Sigmap$ and $\Sigmad$, as in practical specifications of systems with data, process function symbols do take data terms as their parameters.

Terms are built using variables and function symbols by respecting their domains of definition. The sets of open process and data terms are denoted by ${\mathbb{T}}(\Sigma_{P})$ and ${\mathbb{T}}(\Sigma_{D})$, respectively.
Disjointness of process and data variables is mostly for notational convenience.
Function symbols from the process signature are typically denoted by $f_{P}$, $g_{P}$, $f_{P_{i}}$ and $g_{P_{i}}$. Process terms are typically denoted by $t_{P}, t'_{P},$ and $t_{P_{i}}$.
Function symbols from the data signature are typically denoted by $f_{D}, f'_{D}$ and $f_{D_{i}}$, and data terms are typically denoted by $t_{D}, t'_{D},$ and $t_{D_{i}}$.
The sets of closed process and data terms are denoted by $T(\Sigmap)$ and $T(\Sigmad)$, respectively. Closed process and data terms are typically denoted by $p,q,p',p_i,p'_i$ and $d,e,d',d_i,d'_i$, respectively.
We denote process and data substitutions by $\sigma$, $\sigma'$, and $\xi$, $\xi'$, respectively.
We call substitutions $\sigma:\Vp\to {\mathbb{T}}(\Sigma_{P})$ process substitutions and $\xi:\Vd\to{\mathbb{T}}(\Sigma_{D})$ data substitutions.
A substitution replaces a variable in an open term with another (possibly open) term.
Notions of open and closed and the concept of substitution are lifted to formulae in the natural way.

\begin{defn}[Transition System Specification]
Consider a signature $\Sigma$ and a set of labels $L$ (with typical members $l, l', l_0, \ldots$).
A \emph{positive transition formula} is a triple $(t,l,t')$, where $t, t' \in {\mathbb{T}}(\Sigma)$ and $l \in L$, written $t \trans{l} t'$, with the intended meaning: process $t$ performs the action labeled as $l$ and becomes process $t'$.

A \emph{transition rule} is defined as a tuple $(H,\alpha)$, where $H$ is a set of formulae and $\alpha$ is a formula. The formulae from $H$ are called \emph{premises} and the formula $\alpha$ is called the \emph{conclusion}.
A transition rule is mostly denoted by $\sosrule{H}{\alpha}$ and has the following generic shape:
\[
\DedRule{d}~\sosrule{\{t_i \trans{l_{ij}} t_{ij}  \mid  i \in I, j \in J_{i}\}}{t \trans{l} t'}
,\]
where $I, J_{i}$ are sets of indexes, $t, t',t_{i}, t_{ij} \in {\mathbb{T}}(\Sigma)$, and $l_{ij} \in L$.
A \emph{transition system specification} (abbreviated TSS) is a tuple $(\Sigma, L,
{\cal R})$ where $\Sigma$ is a signature, $L$ is a set of labels, and ${\cal R}$ is a set of transition rules of the provided shape.

We extend the shape of a transition rule to handle process terms paired with data terms in the following manner:

\[\DedRule{d'}~\sosrule{\{(t_{P_{i}},t_{D_{i}}) \trans{l_{ij}} (t_{P_{ij}}, t_{D_{ij}}) \mid i \in I, j \in J_{i}\}}{(t_{P}, t_{D}) \trans{l} (t'_{P}, t'_{D})},\]
where $I, J_{i}$ are index sets, $t_{P}, t'_{P},t_{P_{i}}, t_{P_{ij}} \in {\mathbb{T}}(\Sigma_{P})$, $t_{D}, t'_{D},t_{D_{i}}, t_{D_{ij}} \in {\mathbb{T}}(\Sigma_{D})$, and $l_{ij} \in L$.
A \emph{transition system specification with data} is a triple ${\cal T} = (\Sigma_{P} \cup \Sigma_{D}, L,
{\cal R})$ where $\Sigma_{P}$ and $\Sigma_{D}$ are process and data signatures respectively, $L$ is a set of labels, and ${\cal R}$ is a set of transition rules handling pairs of process and data terms.
\end{defn}

\begin{defn}
Let ${\cal T}$ be a TSS with data. A \emph{proof} of a formula $\phi$ from ${\cal T}$ is an upwardly branching tree whose nodes are labelled by formulas such that

\begin{enumerate}
\item the root node is labelled by $\phi$, and
\item if $\psi$ is the label of a node $q$ and the set $\{\psi_{i} \mid i \in I\}$ is the set of labels of the nodes directly above $q$, then there exist a deduction rule $\sosrule{\{\chi_{i} \mid i \in I\}}{\chi}$, a process substitution $\sigma$, and a data substitution $\xi$ such that the application of these substitutions to $\chi$ gives the formula $\psi$, and for all $i \in I$, the application of the substitutions to $\chi_{i}$ gives the formula $\psi_{i}$.
\end{enumerate}

Note that by removing the data substitution $\xi$ from above we obtain the definition for proof of a formula from a standard TSS.
The notation ${\cal T} \vdash \phi$ expresses that there exists a proof of the formula $\phi$ from the TSS (with data) ${\cal T}$. Whenever ${\cal T}$ is known from the context, we will write $\phi$ directly instead of ${\cal T} \vdash \phi$.

\end{defn}


\subsection{Bisimilarity}
In this paper we use two notions of equivalence over processes, one for standard transition system specifications and one for transition system specifications with data. Stateless bisimilarity is the natural counterpart of strong bisimilarity, used in different formalisms such as \cite{Beek06,Beek07,Bergstra07,Galpin12}.

\begin{defn}[Strong Bisimilarity\ \cite{Park81}]\label{def:strongbisim}
Consider a TSS ${\cal T} = (\Sigma_{P}, L, {\cal R})$. A relation $R \subseteq T(\Sigmap) \times T(\Sigmap)$ is a \emph{strong bisimulation} if and only if it is symmetric and $\forall_{p, q} ~(p,q) \in R \Rightarrow (\forall_{l,p'} ~p \trans{l} p' \Rightarrow \exists_{q'} ~ q \trans{l} q' \land (q, q') \in R)$. Two closed terms $p$ and $q$  are \emph{strongly bisimilar}, denoted by $p \bisimpar{{\cal T}} q$ if there exists a strong bisimulation relation $R$ such that $(p,q) \in R$.
\end{defn}

\begin{defn}[Stateless Bisimilarity\ \cite{Mousavi05-IC}]
Consider a TSS with data ${\cal T} = (\Sigma_{P} \cup \Sigma_{D}, L, {\cal R})$. A relation $\Rsb \subseteq T(\Sigmap) \times T(\Sigmap)$ is a \emph{stateless bisimulation} if and only if it is symmetric and $\forall_{p, q} ~(p,q) \in \Rsb \Rightarrow \forall_{d, l, p',d'} ~(p,d) \trans{l} (p',d') \Rightarrow \exists_{q'} ~(q,d) \trans{l} (q',d') \land (p', q') \in \Rsb$. Two closed process terms $p$ and $q$  are \emph{stateless bisimilar}, denoted by $p \bisimsbpar{{\cal T}} q$, if there exists a stateless bisimulation relation $\Rsb$ such that $(p,q) \in \Rsb$.
\end{defn}

\subsection{Rule Formats for Algebraic Properties}
As already stated, the literature on rule formats guaranteeing algebraic properties is extensive. For the purpose of this paper we show the detailed line of reasoning only for the commutativity of binary operators, while, for readability, we refer to the corresponding papers and theorems for the other results in Section~\ref{sec:cases}.

\begin{defn}[Commutativity]
Given a TSS and a binary process operator $f$ in its process signature, $f$ is called
{\em commutative w.r.t.\ $\sim$}, if the following equation is sound w.r.t.\
$\sim$:
\[
f(x_0, x_1) = f(x_1, x_0).
\]
\end{defn}

\begin{defn}[\label{def::commform}Commutativity format\ \cite{BEATCS09}]
A transition system specification over signature $\Sigma$ is in
\commform\ format with respect to a set of binary function symbols
$\COMM \subseteq \Sigma$ if all its $f$-defining transition rules with
$f \in \COMM$ have the following form
\[
\DedRule{c}~\sosrule{\{x_j \trans{l_{ij}} y_{ij}  \mid  i \in I\}}{f(x_{0}, x_1) \trans{l} t}
\]
where $j \in \{0,1\}$, $I$ is an arbitrary index set, and variables
appearing in the source of the conclusion and target of the premises
are all pairwise distinct. We denote the set of premises of
\DedRule{c} by $H$ and the conclusion by $\alpha$. Moreover, for each such rule,
there exist a transition rule \DedRule{c'} of the following form in
the transition system specification
\[
\DedRule{c'}~\sosrule{H'}{f(x'_{0}, x'_{1}) \trans{l} t'}
\]
and a bijective mapping (substitution) $\hbar$ on variables such that
\begin{itemize}
\item $\hbar(x'_0) = x_1$ and  $\hbar(x'_1) = x_0$,
\item $\hbar(t') \simcom t$ and
\item $\hbar(h') \in H$, for each ${h' \in H'}$,
\end{itemize}
where $\simcom$ means equality up to swapping of arguments of operators in $\COMM$ in any context.
Transition rule \DedRule{c'} is called the {\em commutative mirror} of \DedRule{c}.
\end{defn}

\begin{thm}[Commutativity for \textsf{comm-form}\ \cite{BEATCS09}]
\label{thm:commform}
If a transition system specification is in \commform\ format with respect to a set of operators $\COMM$, then all operators in $\COMM$ are commutative with respect to strong bisimilarity.
\end{thm}

\subsection{Sound and ground-complete axiomatizations}
\label{subsec:axiomatizations}

In this section we recall several key aspects presented in \cite{Aceto94}, where the authors provide a procedure for converting any GSOS language definition that disjointly extends the language for synchronization trees to a finite complete equational axiom system which characterizes strong bisimilarity over a disjoint extension of the original language. It is important to note that we work with the  GSOS format because it guarantees that bisimilarity is a congruence and that the transition relation is finitely branching \cite{Bloom95}.
For the sake of simplicity, we confine ourselves to the positive subset of the GSOS format; we expect  the generalization to the full GSOS format to be straightforward.

\begin{defn}[Positive GSOS rule format]
\label{def:gsos}
Consider a process signature $\Sigmap$. A \emph{positive GSOS rule} $\rho$ over $\Sigmap$ has the shape:
\[
\DedRule{g}~\sosrule{\{x_i \trans{l_{ij}} y_{ij}  \mid  i \in I, j \in J_{i}\}}{f(x_{1}, \ldots, x_{n}) \trans{l} C[\vec{x}, \vec{y}]}
,\]
where all variables are distinct, $f$ is an operation symbol form $\Sigmap$ with arity $n$, $I \subseteq \{1, \ldots, n\}$, $J_{i}$ finite for each $i \in I$, $l_{ij}$ and $l$ are labels standing for actions ranging over a given set denoted by $L$, and $C[\vec{x}, \vec{y}]$ is a $\Sigmap$-context with variables including at most the $x_{i}$'s and $y_{ij}$'s.
\end{defn}

A finite tree term $t$ is built according to the following grammar:
\[
t {\,\,::=\,\,} \nul \mid l.t ~ (\forall l \in {L}) \mid t + t.
\]


We denote this signature by $\Sigma_{\textnormal{BCCSP}}$. Intuitively, $\nul$ represents a process that does not exhibit any behaviour, $s + t$ is the nondeterministic choice between the behaviours of $s$ and $t$, while $l.t$ is a process that first performs action $l$ and behaves like $t$ afterwards. The operational semantics that captures this intuition is given by the rules of BCCSP \cite{Glabbeek01}:

\[
\sosrule{}{l.x \xrightarrow{l} x}
\qquad
\sosrule{x \xrightarrow{l} x'}{x + y \xrightarrow{l} x'}
\qquad
\sosrule{y \xrightarrow{l} y'}{x + y \xrightarrow{l} y'}\,.
\]

\begin{defn}[\label{def::axiom}Axiom System]
An {\em axiom} (or {\em equation}) {\em system} $E$ over a signature $\Sigma$ is a set of
equalities of the form $t = t'$, where $t,t' \in
\mathbb{T}({\Sigma})$. An equality $t = t'$, for some $t,t' \in
\mathbb{T}({\Sigma})$, is derivable from $E$, denoted by $E \vdash t = t'$,
if and only if it is in the smallest congruence relation over
$\Sigma$-terms induced by the equalities in $E$.

\end{defn}
We consider the axiom system $E_{\textnormal{BCCSP}}$ which consists of the following axioms:

\begin{center}
\begin{tabular}{r@{\hspace{3pt}}c@{\hspace{3pt}}lr@{\hspace{20pt}}r@{\hspace{3pt}}c@{\hspace{3pt}}lr}
$x + y$ & = & $y + x$ & ~ &
$x + x$ & = & $x$ & ~
\\[1ex]
$(x + y) + z$ & = & $x + (y + z)$ & ~ &
$x + \nul$ & = & $x$\,. &
\end{tabular}
\end{center}

\begin{thm}[\cite{Hennessy85}]
\label{thm:soundness_completeness}
$E_{\textnormal{BCCSP}}$ is sound and ground-complete for bisimilarity on $T(\Sigma_{\textnormal{BCCSP}})$. That is, it holds that
$E_{\textnormal{BCCSP}} \vdash p = q \text{ if, and only if, } p \bisimpar{\BCCSP} \,q$ for any two ground terms $p$ and $q \in T(\Sigma_{\textnormal{BCCSP}})$.
\end{thm}

\begin{defn}[Disjoint extension]
\label{def:disjExt}
A {GSOS} system $G'$ is a disjoint extension of a {GSOS} system $G$, written $G \sqsubseteq G'$, if the signature and the rules of $G'$ include those of $G$, and $G'$ does not introduce new rules for operations in $G$.
\end{defn}

In \cite{Aceto94} it is elaborated how to obtain an axiomatization for a GSOS system $G$ that disjointly extends $\textnormal{BCCSP}$. For technical reasons the procedure involves initially transforming $G$ into a new system $G'$ that conforms to a restricted version of the GSOS format, named \emph{smooth and distinctive}.
We avoid presenting this restricted format, as the method proposed in Section~\ref{sec:axiomatizing} allows us to obtain the axiomatization without the need to transform the initial system $G$.

\section{Currying Data\label{sec:currying}}

We apply the process of currying \cite{springerlink:10.1023/A:1010000313106} known from functional programming to factor out the data from the source and target of transitions and enrich the label to a triple capturing the data flow of the transition. This shows that, for specifying behaviour and data of dynamic systems, the data may be freely distributed over states (as part of the process terms) or system dynamics (action labels of the transition system), providing a natural correspondence between the notions of stateless bisimilarity and strong bisimilarity.
An essential aspect of our approach is that the process of currying is a syntactic transformation defined on transition system specifications (and not a semantic transformation on transition systems); this allows us to apply meta-theorems from the meta-theory of SOS and obtain semantic results by considering the syntactic shape of (transformed) SOS rules.

\begin{defn}[Currying and Label Closure]
\label{def:currying}
Consider the TSS with data ${\cal T} = (\Sigma_{P} \cup \Sigma_{D}, L, \mathcal{R})$ and transition rule
 $\rho \in {\cal R}$ of the shape
$\rho = \sosrule{\{(t_{P_{i}},t_{D_{i}}) \trans{l_{ij}} (t_{P_{ij}}, t_{D_{ij}}) \mid i \in I, j \in J_{i}\}}{(t_{P}, t_{D}) \trans{l} (t'_{P}, t'_{D})}$.

The \emph{curried version} of $\rho$ is the rule $\rho^{c} = \sosrule{\{t_{P_{i}} \trans{(t_{D_i}, l_{ij}, t_{D_{ij}})} t_{P_{ij}} \mid i \in I, j \in J_{i}\}}{t_{P} \trans{(t_{D}, l, t'_{D})} t'_{P}}$.
We further define
${\cal R}^{c} = \{
\rho^{c} \mid
\rho \in \mathcal{R}$\}
and  $L^{c} = \{(t_{D}, l, t'_{D}) \mid l \in L,~~t_{D},t'_{D} \in {\mathbb{T}}(\Sigmad)\}$.
The \emph{curried version} of ${\cal T}$ is defined as ${\cal T}^{c} = (\Sigma_{P}, L^{c},{\cal R}^{c})$.

By $\rho^{c}_{\xi} = \sosrule{\{t_{P_{i}} \trans{(\xi(t_{D_i}), l_{ij}, \xi(t_{D_{ij}}))} t_{P_{ij}} \mid i \in I, j \in J_{i}\}}{t_{P} \trans{(\xi(t_{D}), l, \xi(t'_{D}))} t'_{P}}$ we denote the \emph{closed label version} of $\rho^{c}$ with respect to the closed data substitution $\xi$.
By $\cl(\rho^{c})$ we denote the set consisting of all closed label versions of $\rho^{c}$, i.e. $\cl(\rho^{c}) = \{\rho^{c}_{\xi} \mid \rho^{c} \in {\cal R}^{c}, \xi \text{ is a closed data substitution}\}$. We further define $\cl({\cal R}^{c}) = \{\cl(\rho^{c}) \mid \rho^{c} \in {\cal R}^{c}\}$ and $\cl(L^{c}) = \{(\xi(t_{D}), l, \xi(t'_{D})) \mid (t_{D}, l, t'_{D}) \in L^{c}, \xi \text{ is a closed data substitution}\}$. The \emph{closed label version} of ${\cal T}^{c}$ is $\cl({\cal T}^{c}) = (\Sigma_{P}, \cl(L^{c}),\cl({\cal R}^{c}))$.
\end{defn}

Our goal is to reduce the notion of stateless bisimilarity between two closed process with data terms to strong bisimilarity by means of currying the TSS with data and closing its labels. The following theorem states how this goal can be achieved.

\begin{thm}
\label{thm:sliffsb}
Given a TSS ${\cal T}=(\Sigma,L, D)$ with data, for each two closed process terms $p, q \in T(\Sigmap)$, $p \bisimsbpar{{\cal T}} q$ if, and only if, $p \bisimpar{\cl({\cal T}^{c})} q$.
\end{thm}


\section{Axiomatizing GSOS with Data\label{sec:axiomatizing}}

In this section we provide an axiomatization schema for reasoning about stateless bisimilarity. We find it easier to work directly with curried systems instead of systems with data because this allows us to adapt the method introduced in \cite{Aceto94} by considering the set of more complex labels that integrate the data, as presented in Section~\ref{sec:currying}.

{It is important to note that we present the schema by considering that the signature for data terms, $\Sigma_{D}$, consists only of a finite set of constants.
However, as we foresee a future extension to a setting with arbitrary data terms, we choose to use the notation for arbitrary data terms instead of the one for constants in some of the following definitions.
}

BCCSP is extended to a setting with data, \BCCSPD.
This is done by adding to the signature for process terms $\Sigma_{\BCCSP}$ two auxiliary operators for handling the store, named {\chk} and {\upd}, obtaining a new signature, $\Sigma_{\BCCSPD}$.
Terms over $\Sigma_{\BCCSPD}$ are build according to the following grammar:
\[
t_{P} {\,\,::=\,\,} \nul \mid l.t_{P} ~~~ \forall_{l \in {L}}  \mid \chk(t_{D}, t_{P}) \mid \upd(t_{D}, t_{P}) \mid t_{P} + t_{P} .
\]
Intuitively, operation $\chk(t_{D}, t_{P})$ makes sure that, before executing an initial  action from $t_{P}$, the store has the value $t_{D}$, and $\upd(t_{D}, t_{P})$ changes the store value to $t_{D}$ after executing an initial action of process $t_{P}$. The prefix operation does not affect the store. We directly provide the curried set of rules defining the semantics of $\BCCSPDc$.
\[
\sosrule{}{l.x_{P} \xrightarrow{(x_{D},l,x_{D})} x_{P}}
\qquad
\sosrule{x_{P} \xrightarrow{(x_{D},l,x'_{D})} x'_{P}}{\chk(x_{D}, x_{P}) \xrightarrow{(x_{D},l,x'_{D})} x'_{P}}
\qquad
\sosrule{x_{P} \xrightarrow{(x_{D},l,x'_{D})} x'_{P}}{\upd(y_{D}, x_{P}) \xrightarrow{(x_{D},l,y_{D})} x'_{P}}
\]
\[
\sosrule{x_{P} \xrightarrow{(x_{D},l,x'_{D})} x'_{P}}{x_{P} + y_{P} \xrightarrow{(x_{D},l,x'_{D})} x'_{P}}
\qquad
\sosrule{y_{P} \xrightarrow{(x_{D},l,x'_{D})} y'_{P}}{x_{P} + y_{P} \xrightarrow{(x_{D},l,x'_{D})} y'_{P}}\,.
\]

Definition~\ref{def:strongbisim} can easily be adapted to the setting of SOS systems with data that are curried.

\begin{defn}
\label{def:bisimdata}
Consider a TSS ${\cal T} = (\Sigmap \cup \Sigmad, {L}, {\cal R})$, which means that ${\cal T}^{c} = (\Sigmap, {L^{c}}, {\cal R}^{c})$. A relation $R \subseteq T(\Sigmap) \times T(\Sigmap)$ is a \emph{strong bisimulation} if and only if it is symmetric and $\forall_{p, q} ~(p,q) \in R \Rightarrow \forall_{d, l, d',p'} ~ p \trans{(d, l, d')} p' \Rightarrow \exists_{ q'} ~ q \trans{(d, l, d')} q' \land (q, q') \in R $.
Two closed terms $p$ and $q$  are \emph{strongly bisimilar}, denoted by $p \bisimpar{{\cal T}^{c}} q$ if there exists a strong bisimulation relation $R$ such that $(p,q) \in R$.
\end{defn}

The axiomatization $E_{\BCCSPDc}$ of strong bisimilarity over $\BCCSPDc$, which is to be proven sound and ground-complete in the remainder of this section, is given below:
\[
\begin{array}{lcllr}
x_P + y_P &=&  y_P + x_P  && \mbox{(n-comm)}\\
x_P + (y_P + z_P)  &=&  (x_P + y_P) + z_P  && \mbox{(n-assoc)}\\
x_P + x_P   &=&  x_P && \mbox{(n-idem)}\\
x_P + \nul &=&  x_P  && \mbox{(n-zero)}\\
\chk(x_{D}, x_P + y_P) &=& \chk(x_{D}, x_P) + \chk(x_{D}, y_P) && \mbox{(nc)}\\
\upd(x_{D}, x_P + y_P) &=& \upd(x_{D}, x_P) + \upd(x_{D}, y_P) && \mbox{(nu)}\\
\chk(x_{D}, \upd (y_{D}, x_P) ) &=& \upd(y_{D}, \chk (x_{D}, x_P) ) && \mbox{(cu)} \\
\upd(x_{D},\upd(y_{D}, x_P) ) &=& \upd(x_{D}, x_P) && \mbox{(uu)}\\
\chk(d, \chk (d, x_P) ) &=& x_P \hspace{3ex} (\forall d \in \Sigma_{D}) && \mbox{(cc)} \\
\chk(d, \chk (d', x_P) ) &=& \nul \hspace{4.58ex} (\forall d,d' \in \Sigma_{D}, d \not= d') && \mbox{(cc')} \\
l.x_P &=& \sum_{d \in \Sigma_{D}} \upd (d, \chk (d, l.x_P) ) && \mbox{(lc)} 
\end{array}
\]

Recall that $\Sigma_{D}$ is a finite set of constants, and, therefore, the right hand side of axiom (lc) has a finite number of summands.

The following theorem is proved in the standard fashion.
\begin{thm}[Soundness]
For each two terms $s,t$ in $\mathbb{T}(\Sigma_{\BCCSPDc})$ it holds that if $E_{\BCCSPDc} \vdash s = t$ then $s \bisimpar{\BCCSPDc} t$.
\end{thm}

We now introduce the concept of terms in \emph{head normal form}, which is essential for proving the completeness of axiom systems.

\begin{defn}[Head Normal Form]
\label{def:hnf}
Let $\Sigmap$ be a signature such that $\Sigma_{\BCCSPDc} \subseteq \Sigmap$.
A term $t$ in ${\mathbb T}(\Sigmap)$ is in \emph{head normal form} (for short, h.n.f.) if
\[
t = {\sum_{i \in I} \upd(t'_{Di}, \chk(t_{Di},l_i.t_{Pi}))},\]
where, for every $i \in I$, $t_{Di}, t'_{Di} \in \mathbb{T}(\Sigmad)$, $t_{Pi} \in \mathbb{T}(\Sigmad)$, $l_{i} \in L$.
 The empty sum $(I = \emptyset)$ is denoted by the deadlock constant $\nul$.
\end{defn}

\begin{lem}[Head Normalization]
\label{lem:tgc}
For any term $p$ in ${T}(\Sigma_{\BCCSPDc})$, there exists $p'$ in ${T}(\Sigma_{\BCCSPDc})$ in h.n.f. such that $E_{\BCCSPDc} \vdash p = p'.$ 
\end{lem}
\begin{proof}
By induction on the number of symbols appearing in $p$.
\remarkDG{Should we better call this ``By structural induction on $p$.''?}
We proceed with a case distinction on the head symbol of  $p$.

\newcommand{\eqdef}[1]{\hspace{1ex}{\stackrel{#1}{=}}\hspace{1ex}}

\noindent
\emph{Base case}
\begin{itemize}
\item $p = \nul$; this case is vacuous, because $p$ is already in h.n.f.
\end{itemize}

\noindent
\emph{Inductive step cases}
\begin{itemize}
\item
$p$ is of the shape $l.p'$; then

\hspace{-1.5ex}
$
\begin{array}{l}
 p \eqdef{\textnormal{def. } p} l.p' \eqdef{\textnormal{(lc)}} 
 \sum_{d \in T(\Sigma_{D})} \upd (d, \chk (d, l.p') ) \textnormal{, which is in h.n.f.}
\end{array}
$

\item $p$ is of the shape $\chk(d'', p'')$; then

\hspace{-1.5ex}
$
\begin{array}{l}
p \eqdef{\textnormal{def. } p}
\chk(d'', p'')  \eqdef{\textnormal{ind. hyp. }} \\
\chk(d'', \sum_{i \in I} \upd (d'_{i}, \chk (d_{i}, l_{i}.p'_{i}))) \eqdef{\textnormal{(nc)}} 
\sum_{i \in I} \chk(d'',\upd(d'_{i}, \chk(d_{i}, l_{i}.p'_{i}))) \eqdef{\textnormal{(cu)}} \\
\sum_{i \in I} \upd(d'_{i}, \chk(d'', \chk(d_{i}, l_{i}.p'_{i}))) \eqdef{\textnormal{(cc,cc')}}
\sum_{i \in I, d_{i} = d''} \upd(d'_{i}, \chk(d_{i}, l_{i}.p'_{i})), \\\textnormal{which is in h.n.f.}
\end{array}
$

\item $p$ is of the form $\upd(d'', p'')$; then

\hspace{-1.5ex}
$
\begin{array}{l}
p \eqdef{\textnormal{def. } p}
\upd(d'', p'')  \eqdef{\textnormal{ind. hyp. }}
\upd(d'', \sum_{i \in I} \upd (d'_{i}, \chk (d_{i}, l_{i}.p'_{i}))) \eqdef{\textnormal{(nu)}} \\
\sum_{i \in I}\upd(d'', \upd (d'_{i}, \chk (d_{i}, l_{i}.p'_{i})))  \eqdef{\textnormal{(uu)}} 
\sum_{i \in I}\upd(d'', \chk (d_{i}, l_{i}.p'_{i})), \\
\textnormal{which is in h.n.f.}
\end{array}
$

\item $p$ is of the form $p_0 + p_1$; then

$
\begin{array}{l}
p \eqdef{\textnormal{def. } p}
p_0 + p_1 \eqdef{\textnormal{ind. hyp.}}
\sum_{i \in I} \chk(d'',\upd(d'_{i}, \chk(d_{i}, l_{i}.p'_{i}))) ~+ \\
\sum_{j \in J} \chk(d'',\upd(d'_{j}, \chk(d_{j}, l_{j}.p'_{j}))) = \\
\sum_{k \in I \cup J} \chk(d'',\upd(d'_{k}, \chk(d_{k}, l_{k}.p'_{k}))) \textnormal{, which is in h.n.f.}
\end{array}
$

\end{itemize}
\end{proof}

\begin{thm}[Ground-completeness]
For each two closed terms $p,q \in {T}(\Sigma_{\BCCSPDc})$, it holds that if $p \bisimpar{\BCCSPDc} q$, then $E_{\BCCSPDc} \vdash p = q$.
\end{thm}

\begin{proof}
We assume, by Lemma~\ref{lem:tgc} that $p,q$ are in h.n.f., define the function $\textit{height}$ as follows:

\[\textit{height}(p)=
\left\{
\begin{array}{ll}
0 & \text{ if } p = \nul  \\
1 + \max(\textit{height}(p_1), \textit{height}(p_2)) & \text{ if } p = p_1 + p_2 \\
1 + \textit{height}(p') & \text{ if } p = \upd(d', \chk(d, l.p')),
\end{array}
\right.
\]
and prove the property by induction on $M = \max(\textit{height}(p),\textit{height}(q))$.

\medskip
\noindent
\emph{Base case} ($M = 0$) This case is vacuous, because  $p = q = \nul$, so $E_{\BCCSPDc} \vdash p = q$.

\medskip
\noindent
\emph{Inductive step case} ($M > 0$) We prove $E_{\BCCSPDc} \vdash p = q + p$ by arguing that every summand of $q$ is provably equal to a summand of $p$. Let $\upd(d', \chk(d, l.q'))$ be a summand of $q$. By applying the rules defining {\BCCSPDc}, we derive $q \trans{(d,l,d')} q'$. As $q \bisimpar{\BCCSPDc} p$ holds, it has to be the case that $p \trans{(d,l,d')} p'$ and $q' \bisimpar{\BCCSPDc} p'$ hold. As $\max(\textit{height}(q'),\textit{height}(p')) < M$, from the inductive hypothesis it results that $E_{\BCCSPDc} \vdash q' = p'$, hence $\upd(d', \chk(d, l.q'))$ is provably equal to $\upd(d', \chk(d, l.p'))$, which is a summand of $p$.

It follows, by symmetry, that $E_{\BCCSPDc} \vdash q = p + q$ holds, which ultimately leads to the fact that $E_{\BCCSPDc} \vdash p = q$ holds.
\end{proof}

{
Consider a TSS with data ${\cal T} = (\Sigma_{P} \cup \Sigma_{D}, L, \mathcal{R})$. For an operation $f \in \Sigma_{P}$, we denote by $\mathcal{R}_{f}$ the set of all rules defining $f$. All the rules in $\mathcal{R}_{f}$ are in the GSOS format extended with the data component.
For the simplicity of presenting the axiomatization schema, we assume that $f$ only has process terms as arguments, baring in mind that adding data terms is trivial.
}

When given a signature $\Sigma_{P}$ that includes $\Sigma_{\BCCSPDc}$, the purpose of an axiomatization for a term $p \in T(\Sigma_{P})$ is to derive another term $p'$ such that $p \bisimpar{{\cal T}^{c}} \,p'$ and $p' \in T(\Sigma_{\BCCSPDc})$.

\begin{defn}[Axiomatization schema]\label{def:axiomatization} Consider a TSS ${\cal T}^{c} = (\Sigma_{P}, L^{c}, {\cal R}^{c})$ such that $\textnormal{\BCCSPDc} \sqsubseteq {\cal T}^{c}$. By $E_{{\cal T}^{c}}$ we denote the axiom system that extends $E_{\BCCSPDc}$ with the following axiom schema for every operation $f$ in $\cal T$, parameterized over the vector of closed process terms $\vec{p}$ in h.n.f.:

\begin{center}
\noindent
$f(\vec{p}) = \sum \left\{\upd(d', \chk(d, l.C[\vec{p},\vec{y_{P}}])) \,\,\Bigg|\,\, \rho = \sosrule{H}{f(\vec{p}) \xrightarrow{(d, l, d')} C[\vec{p},\vec{q}]} \in cl({\cal R}^{c}_{f}) \textnormal{ ~and }
 ~\checkmark(\vec{p}, \rho)
 \right\}$,
\end{center}

\noindent
\hspace{5ex} where $\checkmark$ is defined as
~~\(
\checkmark(\vec{p}, \rho) = \bigwedge_{{p_{k}} \in \vec{p}} ~\checkmark'(p_{k}, k, \rho),\)

\noindent
\hspace{5ex}
\(
\begin{array}{@{}l@{}l}
\text{and }~~
\checkmark' &\left( {p_{k}}, k, \sosrule{\{x_{P_i} \trans{({d_{i}}, l_{ij}, {d'_{ij}})} y_{P_{ij}}  \mid  i \in I, j \in J_{i}\}}{f(\vec{p}) \xrightarrow{(d, l, d')} C[\vec{p},\vec{q}]} \right) =\\[3ex]
&~\textnormal{if } k \in I \textnormal{ then } \forall_{j \in J_{k}}~~ \exists_{p', p''}~~ {E_{\BCCSPDc} \vdash {p_{k}} = \upd({d'_{kj}}, \chk({d_{k}}, l_{kj}.{p'})) + {p''}}.
\end{array}
\)

\end{defn}

Intuitively, the axiom transforms $f(\vec{p})$ into a sum of closed terms covering all its execution possibilities. We iterate, in order to obtain them, through the set of $f$-defining rules and check if $\vec{p}$ satisfies their hypotheses by using the meta-operation $\checkmark$. $\checkmark$ makes sure that, for a given rule, every component of $\vec{p}$ is a term with enough action prefixed summands satisfying the hypotheses associated to that component. Note that the axiomatization is built in such a way that it always derives terms in head normal form. Also note that the sum on the right hand side is finite because of our initial assumption that the signature for data is a finite set of constants.

The reason why we conceived the axiomatization in this manner is of practical nature.
Our past experience shows that this type of schemas may bring terms to their normal form faster than finite axiomatizations.
Aside this, we do not need to transform the initial system, as presented in \cite{Aceto94}.

\begin{thm}
Consider a TSS ${\cal T}^{c} = (\Sigma_{P}, L^{c}, {\cal R}^{c})$ such that ${\BCCSPDc} \sqsubseteq {\cal T}^{c}$. $E_{{\cal T}^{c}}$ is sound and ground-complete for strong bisimilarity on $T(\Sigma_{P})$.
\end{thm}
\begin{proof}
It is easy to see that, because of the head normal form of the right hand side of every axiom, the completeness of the axiom schema reduces to the completeness proof for bisimilarity on $T(\Sigma_{\BCCSPDc})$.

In order to prove the soundness, we denote, for brevity, the right hand side of the schema in Definition~\ref{def:axiomatization} by \textit{RHS}. 

Let us first prove that if $f(\vec{p})$ performs a transition then it can be matched by \textit{RHS}. Consider a rule $\rho \in cl({\cal R}^{c}_{f})$ that can be applied for $f(\vec{p})$: $\rho = \sosrule{\{x_i \trans{({d_{i}}, l_{ij}, {d_{ij}})} y_{ij}  \mid  i \in I, j \in J_{i}\}}{f(\vec{x}) \trans{(d, l, d')} C[\vec{x}, \vec{y}]}$.
Then $f(\vec{p}) \trans{(d, l_{ij}, d')} C[\vec{p}, \vec{q}]$ holds and, at the same time, all of the rule's premises are met.
This means that $p_{i}$ is of the form $\sum_{j \in J_{i}} \upd({d_{ij}}, \chk({d_{i}}, l_{ij}.p_{ij})) + p'$ for some $p'$ and $p_{ij}$'s. It is easy to see that all the conditions for $\checkmark$ are met, so $(d, l, d').C[\vec{p},\vec{q}]$ is a summand of \textit{RHS}, and therefore it holds that $\textit{RHS} \trans{({d_{i}}, l_{ij}, {d_{ij}})} C[\vec{p},\vec{q}]$, witch matches the transition from $f(\vec{p})$.

The proof for the fact that $f(\vec{p})$ can match any of the transitions of \textit{RHS} is similar.
\end{proof}

We end this section with the remark that the problem of extending the axiomatization schema to the setting with arbitrary data terms is still open.
The most promising solution we have thought of involves using the infinite alternative quantification operation from \cite{Reniers02}.
This operation would us to define and express head normal forms as (potentially) infinite sums, parameterized over data variables.

\section{Case Study: The Coordination Language Linda \label{sec:cases}}

In what follows we present the semantics and properties of a core prototypical language.

The provided specification defines a structural operational semantics for  the coordination language Linda;
the specification is taken from \cite{Mousavi05-IC} and is a slight adaptation of the original semantics presented in \cite{Brogi98} (by removing structural congruences and introducing a terminating process $\epsilon$).
Process constants (atomic process terms) in this language are $\epsilon$ (for terminating process), $\ask(u)$ and $\nask(u)$ (for checking existence and absence of tuple $u$ in the shared data space, respectively), $\tell(u)$ (for adding tuple $u$ to the space) and $\get(u)$ (for taking tuple $u$ from the space). Process composition operators in this language include nondeterministic choice ($+$), sequential composition ($\Seq$) and parallel composition ($ \parallel $). The data signature of this language consists of a constant $\{\}$ for the empty multiset and a class of unary function symbols $\cup \{u\}$, for all tuples $u$, denoting the union of a multiset with a singleton multiset containing tuple $u$. The operational state of a Linda program is denoted by $\State{p}{\sdspace}$ where $p$ is a process term in the above syntax and $\sdspace$ is a multiset modeling the shared data space.

The transition system specification defines one relation $\trans{}$ and one predicate $\downarrow$. Note that $\trans{}$ is unlabeled, unlike the other relations considered so far.
Without making it explicit, we tacitly consider the termination predicate $\downarrow$ as a binary transition relation $\trans{\downarrow}$ with the pair $(x_{P}, x_{D})$, where $x_{P}$ and $x_{D}$ are fresh yet arbitrary process and data variables, respectively.

Below we provide a table consisting of both the original and the curried and closed label  versions of the semantics of Linda on the left and, respectively, on the right.

\begin{longtable}{c@{\qquad}c}
$\nr{1} \sosrule{}{\State{\epsilon}{\sdspace}\downarrow}$ &
$\nr{1_c} \sosrule{}{\epsilon\downarrow}$
\\[4ex]
$\nr{2} \sosrule{}{(\ask(u),{\sdspace \cup \{u\}})\trans{} \State{\epsilon}{\sdspace \cup \{u\}}}$ &
$\nr{2_c} \sosrule{}{\ask(u) \trans{(d\cup \{u\}, -, d\cup \{u\})} \epsilon}$
\\[4ex]
$\nr{3} \sosrule{}{\State{\tell(u)}{\sdspace} \gtrans{}{} \State{\epsilon}{\sdspace \cup \{u\}}}$ &
$\nr{3_c} \sosrule{}{\tell(u) \trans{(d, -, d \cup \{u\})} \epsilon}$
\\[4ex]
$\nr{4} \sosrule{}{\State{\get(u)}{\sdspace \cup \{u\}} \gtrans{}{} \State{\epsilon}{\sdspace}}$ &
$\nr{4_c} \sosrule{}{\get(u) \trans{(d \cup \{u\}, -, d)} \epsilon}$
\\[4ex]
$\nr{5} \sosrule{}{\State{\nask(u)}{\sdspace} \gtrans{}{} \State{\epsilon}{\sdspace}}[u \notin \sdspace]$ &
$\nr{5_c} \sosrule{}{\nask(u) \trans{(d, -, d)} \epsilon}[u \notin d]$
\\[3ex]
$\nr{6} \sosrule{\State{x_{P}}{\sdspace}\downarrow}{\State{x_{P}+y_{P}}{\sdspace}\downarrow}$ ~~ $\nr{7} \sosrule{\State{y_{P}}{\sdspace}\downarrow}{\State{x_{P}+y_{P}}{\sdspace}\downarrow}$ &
$\nr{6_c} \sosrule{x_{P} \downarrow}{x_{P}+y_{P}\downarrow}$ ~~ $\nr{7_c} \sosrule{y\downarrow}{x_{P}+y_{P}\downarrow}$
\\[3ex]
$\nr{8} \sosrule{\State{x_{P}}{\sdspace} \gtrans{}{} \State{x'_{P}}{\sdspace'}}{\State{x_{P}+y_{P}}{\sdspace} \gtrans{}{} \State{x'_{P}}{\sdspace'}}$ &
$\nr{8_c} \sosrule{x_{P} \trans{(d,-,d')} x'_{P}}{x_{P}+y_{P} \trans{(d,-,d')} x'_{P}}$
\\[3ex]
$\nr{9} \sosrule{\State{y_{P}}{\sdspace} \gtrans{}{} \State{y'_{P}}{\sdspace'}}{\State{x_{P}+y_{P}}{\sdspace} \gtrans{}{} \State{y'_{P}}{\sdspace'}}$ &
$\nr{9_c} \sosrule{y_{P} \trans{(d,-,d')} y'_{P}}{x_{P} + y_{P} \trans{(d,-,d')} y'_{P}}$
\\[3ex]
$\nr{10}\sosrule{\State{x_{P}}{\sdspace} \gtrans{}{} \State{x'_{P}}{\sdspace'}}{\State{x_{P} \Seq y_{P}}{\sdspace} \gtrans{}{} \State{x'_{P} \Seq y_{P}}{\sdspace'}}$ &
$\nr{10_c} \sosrule{x_{P} \trans{(d,-,d')} x'_{P}}{x_{P} \Seq y_{P} \trans{(d,-,d')} x'_{P} \Seq y_{P}}$
\\[3ex]
$\nr{11} \sosrule{\State{x_{P}}{\sdspace}\downarrow ~ \State{y_{P}}{\sdspace} \gtrans{}{} \State{y'_{P}}{\sdspace'}}{\State{x_{P} \Seq y_{P}}{\sdspace} \gtrans{}{} \State{y'_{P}}{\sdspace'}}$ &
$\nr{11_c} \sosrule{x_{P} \downarrow \quad y_{P} \trans{(d,-,d')} y'_{P}}{x_{P} \Seq y_{P} \trans{(d,-,d')} y'_{P}}$
\\[3ex]
$\nr{12} \sosrule{\State{x_{P}}{\sdspace}\downarrow ~ \State{y_{P}}{\sdspace}\downarrow}{\State{x_{P} \Seq y_{P}}{\sdspace} \downarrow}$ &
$\nr{12_c} \sosrule{x_{P}\downarrow ~ y_{P}\downarrow}{x_{P} \Seq y_{P} \downarrow}$
\\[3ex]
$\nr{13} \sosrule{\State{x_{P}}{\sdspace} \gtrans{}{} \State{x'_{P}}{\sdspace'}}{\State{x_{P} \parallel y_{P}}{\sdspace} \gtrans{}{} \State{x' \parallel y}{\sdspace'}}$ &
$\nr{13_c} \sosrule{x_{P} \trans{(d,-,d')} x'_{P}}{x_{P} \Par y_{P} \trans{(d,-,d')} x'_{P} \Par y_{P}}$
\\[3ex]
$\nr{14} \sosrule{\State{y_{P}}{\sdspace} \gtrans{}{} \State{y'_{P}}{\sdspace'}}{\State{x_{P} \parallel y_{P}}{\sdspace} \gtrans{}{} \State{x_{P} \parallel y'_{P}}{\sdspace'}}$ &
$\nr{14_c} \sosrule{y_{P} \trans{(d,-,d')} y'_{P}}{x_{P} \Par y_{P} \trans{(d,-,d')} x_{P} \Par y'_{P}}$
\\[3ex]
$\nr{15} \sosrule{\State{x_{P}}{\sdspace}\downarrow ~ \State{y_{P}}{\sdspace}\downarrow}{\State{x_{P}  \parallel y_{P}}{\sdspace}\downarrow}$ &
$\nr{15_c} \sosrule{x_{P} \downarrow ~ y_{P}\downarrow}{x_{P} \parallel y_{P}\downarrow}$
\end{longtable}

In the curried SOS rules, $d$ and $d'$ are arbitrary closed data terms, i.e., each transition rule given in the curried semantics represents a (possibly infinite) number of rules for each and every particular $d,d' \in T(\Sigmad)$. It is worth noting that by using the I-MSOS framework~\cite{DBLP:journals/entcs/MossesN09} we can present the curried system without explicit labels at all as they are propagated implicitly between the premises and conclusion.

Consider transition rules $({6_c})$, $({7_c})$, $({8_c})$, and $({9_c})$; they are the only $+$\,-\,defining rules and
they fit in the commutativity format of Definition~\ref{def::commform}.
It follows from Theorem~\ref{thm:commform} that the equation $x + y = y + x$ is sound with respect to strong bisimilarity in the curried semantics.
Subsequently, following Theorem \ref{thm:sliffsb}, we have that the previously given equation is sound with respect to stateless bisimilarity in the original semantics. (Moreover, we have that $(x_0 + x_1, d) = (x_1 + x_2, d)$ is sound with respect to statebased bisimilarity for all $d \in T(\Sigmad)$.)

Following a similar line of reasoning, we get that $x \Par y = y \Par x$ is sound with respect to stateless bisimilarity in the original semantics.

In addition, we derived the following axioms for the semantics of Linda, using the meta-theorems stated in the third column of the table.
The semantics of sequential composition in Linda is identical to the sequential composition (without data) studied in Example~9 of \cite{Mousavi08-CONCUR}; there, it is shown that this semantics conforms to the \textsc{Assoc-De Simone} format introduced in \cite{Mousavi08-CONCUR} and hence, associativity of sequential composition follows immediately.
Also semantics of nondeterministic choice falls within the scope of the \textsc{Assoc-De Simone} format (with the proposed coding of predicates), and hence, associativity of  nondeterministic choice follows (note that in \cite{Mousavi08-CONCUR} nondeterministic choice without termination rules is treated in Example~1; moreover, termination rules in the semantics of parallel composition are discussed in Section 4.3 and shown to be safe for associativity).
Following a similar line of reasoning associativity of parallel composition follows from the conformance of its rules to the
\textsc{Assoc-De Simone} format of \cite{Mousavi08-CONCUR}.
Idempotence for $+$ can be obtained, because rules $({6_c})$, $({7_c})$ and $({8_c})$, $({9_c})$ are choice rules \cite[Definition 40]{Mousavi12-SCICO} and the family of rules  $({6_c})$ to $({9_c})$ for all data terms $d$ and $d'$ ensure that the curried specification is in idempotence format with respect to the binary operator $+$. The fact that  $\epsilon$ is unit element for $;$ is proved similarly as in \cite{Mousavi11-TCS}, Example~10.

\[
\begin{array}{l | c | l}
\hline
\mbox{Property} & \mbox{Axiom} & \mbox{Meta-Theorem} \\
\hline
\mbox{Associativity for } \Seq & x  \Seq (y \Seq z) = (x \Seq  y) \Seq z & \text{ Theorem 1 of } \cite{Mousavi08-CONCUR}\\  
\mbox{Associativity for } +    & x  + (y + z) = (x + y) + z & \text{ Theorem 1 of } \cite{Mousavi08-CONCUR}\\ 
\mbox{Associativity for } \Par & x  \Par (y \Par z) = (x \Par  y) \Par z & \text{ Theorem 1 of } \cite{Mousavi08-CONCUR}\\  
\mbox{Idempotence for } + & x  + x = x & \text{ Theorem 42 of } \cite{Mousavi12-SCICO} \\  
\mbox{Unit element for } \Seq & \epsilon \Seq x = x & \text{ Theorem 3 of } \cite{Mousavi11-TCS} \\  
\mbox{Distributivity of } + \mbox{ over } \Seq & (x + y) \Seq z = (x \Seq y) + (x \Seq z) & \text{ Theorem 3 of } \cite{Mousavi11-LATA} \\
\hline
\end{array}
\]

We currently cannot derive an axiomatization for Linda because its semantics involves arbitrary data terms, as opposed to a finite number of constants.

\section{Conclusions\label{sec:conc}}

In this paper, we have proposed a generic technique for extending the meta-theory of algebraic properties to SOS with data, memory or store.
In a nutshell, the presented technique allows for focusing on the structure of the process (program) part in SOS rules and ignoring the data terms in order to obtain algebraic properties, as well as, a sound and ground complete set of equations
\remarkDG{Any reason that we do not mention the completeness property?}
w.r.t.\ stateless bisimilarity.
We have demonstrated the applicability of our method by means of the well known coordination language Linda.


It is also worth noting that one can check whether a system is in the process-tyft format presented in \cite{DBLP:conf/lics/MousaviRG04} in order to infer that stateless bisimilarity is a congruence, and if this is the case, then strong bisimilarity over the curried system is also a congruence. Our results are applicable to a large body of existing operators in the literature and make it possible to dispense with several lengthy and laborious soundness proofs in the future.


Our approach can be used to derive algebraic properties that are sound with respect to weaker notions of bisimilarity with data, such as initially stateless and statebased bisimilarity \cite{Mousavi05-IC}. We do expect to obtain stronger results, e.g., for zero element with respect to statebased bisimilarities, by scrutinizing data dependencies particular to these weaker notions.
We would like to study coalgebraic definitions of the notions of bisimilarity with data (following the approach of \cite{Turi97}) and
develop a framework for SOS with data using the  bialgebraic approach.
{Furthermore, it is of interest to check how our technique can be applied to quantitative systems where non-functional aspects like probabilistic choice or stochastic timing is encapsulated as data.}
We also plan to investigate the possibility of automatically deriving axiom schemas for systems whose data component is given as arbitrary terms, instead of just constants.

\paragraph{Acknowledgements.} We thank Luca Aceto, Peter Mosses, and Michel Reniers for their valuable comments on earlier versions of the paper. Eugen-Ioan Goriac is funded by
 the project `Extending and Axiomatizing Structural Operational
 Semantics: Theory and Tools' (nr.~1102940061) of the Icelandic
 Research Fund.

\end{document}